\documentclass[aps, a4paper, reprint]{revtex4-2}

\usepackage{amsmath, amssymb, amsfonts, physics, amsthm}
\usepackage{hyperref}
\hypersetup{
	colorlinks=true,
	linkcolor=blue
}
\usepackage{graphicx,xcolor}
\newtheorem{theorem}{Theorem}
\newtheorem{definition}{Definition}

\newtheorem{remark}{Remark}
\newtheorem{lemma}{Lemma}

\newtheorem{proposition}{Proposition}
\usepackage{blindtext}

\begin{document}
\title{On genuine multipartite entanglement signals}
\date{\today}
\author{Abhijit Gadde}
\email{abhijit@theory.tifr.res.in}
\affiliation{Department of Theoretical Physics, Tata Institute of Fundamental Research, Homi Bhabha Road, Mumbai, 400005.}

\begin{abstract}
We give a general construction of genuinely multipartite entanglement signals from families of lower-partite symmetric local-unitary invariants satisfying a natural compatibility condition. M\"obius inversion on the partition lattice plays a key role in this construction. We show that many examples of multipartite entanglement signals considered in the literature fit naturally into this framework. We also explain how the genuinely multipartite signal can be extracted from a general, not necessarily symmetric, multi-invariant.
\end{abstract}

\maketitle

\section{States and separability}
We will first develop some terminology regarding multi-partite states and their separability. A quantum state $|\psi\rangle$ is called $q$-partite if $|\psi\rangle \in {\cal H}$ where ${\cal H}$ is a tensor product of $q$ factors, ${\cal H}={\cal H}_{A_1}\otimes \ldots \otimes {\cal H}_{A_q}$. The subscript of each  factor is called its party. 
For general discussion we will keep the notation $A_1, \ldots, A_q$ for parties and for specific examples, we will switch to labeling parties as $A, B, C, \ldots$ etc. Separable states are classified by partitions of the set  $X=\{A_1,\ldots, A_q\}$ so it is useful to review the standard terminology regarding the same.

A partition of a set $X$ is a collection of non-empty subsets of $X$ (called blocks) which are disjoint and cover the whole set. We denote the set of  partitions of $X$ as $\Pi_X$. 
We denote partitions i.e. elements of $\Pi_X$ by the greek letters $\pi, \rho$ etc. and denote its blocks as $B_1,\ldots, B_m$. The number of blocks in $\pi$ is denoted as $|\pi|$. The integer partition  of $|X|$, given by the block sizes of $\pi$, arranged in non-decreasing order  is called the type of $\pi$ and is denoted as $\lambda(\pi)$. 
The \emph{finest}  partition is the unique partition of $X$ into singletons i.e. blocks of size $1$, characterized by $|\pi|=|X|$; we denote it by ${\bf 0}$.
The \emph{coarsest}  partition is the unique partition with a single block, characterized by $|\pi|=1$; we denote it by ${\bf 1}$.
An example of a partition of the set  $\{A,B,C,D,E\}$ is $\{\{A,B\},\{C,D,E\}\}$. It is of the type $(2,3)$ and its blocks are $B_1=\{A,B\}$ and $B_2=\{C,D,E\}$. 

A pure $q$-partite state $|\psi\rangle$ is called $\pi$-separable if the state $|\psi\rangle$ factorizes across the blocks of $\pi$ i.e.
\begin{align}
    |\psi\rangle=\bigotimes_{B\in \pi} |\psi_B\rangle, \qquad {\rm with}\quad |\psi_B\rangle\in \bigotimes_{A_a\in B}{\cal H}_{A_a}.
\end{align}
We call the $\pi$-separable state $|\psi\rangle$ 
\begin{itemize}
    \item \emph{completely separable} if $\pi={\bf 0}$. 
    \item \emph{separable} if $|\pi|\geq 2$.
\end{itemize}
We call the state \emph{genuinely multipartite entangled (GME)}  if the state is not separable.
We will often use the notation $|\psi\rangle_\pi$ to denote the separability type of the state.

Let $|\psi\rangle \in \bigotimes_{a=1}^q \mathcal H_{A_a}$ be a $q$-partite pure state. 
Assume that each local Hilbert space further factorizes as
\[
\mathcal H_{A_a} \;=\; \bigotimes_{i=1}^n \mathcal H_{A_a^{(i)}} \qquad (a=1,\dots,q).
\]
We say that $|\psi\rangle$ \emph{admits a layer decomposition} (or \emph{factorizes into layers}) if there exist $q$-partite states
\[
|\psi^{(i)}\rangle \in \bigotimes_{a=1}^q \mathcal H_{A_a^{(i)}} \quad {s.t.}\quad
|\psi\rangle \;=\; \bigotimes_{i=1}^n |\psi^{(i)}\rangle .
\]
The states $\{|\psi^{(i)}\rangle\}_{i=1}^n$ are called the \emph{layers} of $|\psi\rangle$.

We call $|\psi\rangle$ \emph{layerwise-separable} if it admits a layer decomposition for which every layer $|\psi^{(i)}\rangle$ is separable.
In particular, any separable state is layerwise-separable (take $n=1$), whereas the converse need not hold. 
We regard the non-separable layers as the GME part of the state.

In order to probe multipartite entanglement properties of a pure state
$|\psi\rangle \in \bigotimes_{a=1}^q \mathcal H_{A_a}$, we consider functions of $|\psi\rangle$
that are invariant under \emph{local unitary} (LU) transformations,
\begin{align}
    |\psi\rangle \;\longmapsto\; (U_1\otimes \cdots \otimes U_q)\,|\psi\rangle,
\end{align}
where each $U_a$ is a unitary operator acting on $\mathcal H_{A_a}$.

\begin{definition}
    An LU-invariant function $f: \otimes_{a=1}^q  {\cal H}_{A_a}\to {\mathbb C}$ is called additive (multiplicative) if it is additive (multiplicative) under layer decomposition. 
\end{definition}
\begin{definition}[Signal]
    A signal is an LU-invariant function $f:\otimes_{a=1}^q  {\cal H}_{A_a}\to {\mathbb C}$ that vanishes on layerwise-separable states.
\end{definition}
\begin{definition}[Pre-signal]
    A pre-signal is an LU-invariant function $f:\otimes_{a=1}^q  {\cal H}_{A_a}\to {\mathbb C}$ that vanishes on separable states.
\end{definition}
\begin{remark}\label{signal_rel}
    A pre-signal that is additive is a signal.
\end{remark}
\noindent In fact, one can make an even stronger statement. 
\begin{lemma}\label{key_lemma}
    If an LU-invariant function is additive and vanishes for $\pi$-separable state $|\psi\rangle$ of type  $(1,q-1)$ then it is a  signal.
\end{lemma}
\noindent 
This is because if the state layerwise-separable then each layer can be written as a tensor product of layers, each of which is of the type $(1,q-1)$. For example, if a $5$-partite state is separable of the type $(2,3)$ then it can be considered a product of two layers of the type $(1,1,1,2)$ and $(1,1,3)$. Each of these layers are of the type $(1,4)$ as well.
\begin{remark}
    If $f(|\psi\rangle)$ is a signal (pre-signal) then $g(f(|\psi\rangle))$ is also a signal (pre-signal) for any $g:{\mathbb C}\to {\mathbb C}$.
\end{remark}
\noindent 

\subsection*{Signals}
A signal is a  useful indicator for the presence of GME layer of the state. 
A generic local-unitary (LU) invariant function need not respect a layer decomposition: it may combine information
from different local tensor factors and thereby obscure the layered structure. Consequently, a signal must, in effect,
identify a local basis (or, more generally, a local tensor-product structure) in which the state admits a layer
decomposition. A simple way to enforce compatibility with layers is to use Remark~\ref{signal_rel} and work with
pre-signals that are additive. In Section~\ref{general_construct}, we follow this
strategy, in particular Lemma~\ref{key_lemma}, and construct additive signals as linear combinations of logarithms of multi-invariants. Since
multi-invariants are multiplicative, their logarithms are additive. These signals play a crucial
role in identifying the low-energy topological quantum field theory from the ground-state wavefunction \cite{DelZotto:2026fpw}. Additive signals based on multi-entropy and Renyi multi-entropy -- called genuine multi-entropy \cite{Iizuka:2025ioc, Iizuka:2025caq}  -- have been used to diagnose genuine multipartite entanglement in holographic states. 

Signals, however, need not be additive. For instance,
\begin{align}
    f(|\psi\rangle)
    \;:=\;
    1-\sup_{\,|\phi\rangle\in \mathcal L}\, \bigl|\langle\phi|\psi\rangle\bigr|^2
\end{align}
defines a signal, where $\mathcal L$ denotes the set of all layerwise-separable pure states. We will not explore such non-additive signals in this paper.

\subsection*{GME measures}
Signals and pre-signals may be viewed as stepping stones toward genuine multipartite entanglement (GME) measures, which in addition 
satisfy positivity and monotonicity under LOCC on average. Unfortunately, 
\begin{theorem}\label{signal_no}
No signal can serve as a genuine multipartite entanglement (GME) measure.
\end{theorem}
\noindent
We prove Theorem~\ref{signal_no} in Appendix~\ref{no_signal} by showing that, for signals, LOCC monotonicity is
incompatible with positivity on non-layerwise-separable states. In particular, we exhibit a local operation that
maps a layerwise-separable state to a non-layerwise-separable state.

This limitation does not obstruct the construction of GME measures, since a GME measure only needs to be a pre-signal
(in addition to being positive and LOCC monotone). In this sense, pre-signals are more natural objects than signals
from a resource-theoretic point of view.

A positive, LOCC-monotone pre-signal can be obtained in a straightforward way by taking a bipartite entanglement
measure and minimizing it over all bipartitions of a $q$-partite state. However, such a quantity is designed to detect
whether \emph{some} bipartition is weakly entangled and therefore does not faithfully capture the genuinely multipartite
character of entanglement.
It would be interesting to explore if the general construction of pre-signals presented here can help in the construction of interesting GME measures. We leave this problem for future work.

\subsection*{Symmetric LU-invariants}
There is a natural coarse-graining map from $q$-partite states to $m$-partite states, specified by a partition
$\pi\in \Pi_X$ with $|\pi|=m$. Writing $\pi=\{B_1,\ldots,B_m\}$, we group together all parties whose labels lie in the
same block and regard their tensor product as a single effective subsystem. Concretely, a state
$|\psi\rangle\in \bigotimes_{A_a\in X}\mathcal H_{A_a}$ may be viewed as an $m$-partite state
\[
|\psi\rangle\in \bigotimes_{i=1}^m \mathcal H_{B_i},
\qquad\text{where}\qquad
\mathcal H_{B_i}:=\bigotimes_{A_a\in B_i}\mathcal H_{A_a},
\]
with the $i$th party corresponding to the block $B_i$. We denote the associated $\pi$-coarse-graining map by
${\rm grp}_\pi$. For example, if $X=\{A,B,C,D,E\}$ and $\pi=\{\{A,B\},\{C,D,E\}\}$, then ${\rm grp}_\pi$ allows us to
regard the $5$-partite state $|\psi\rangle_{A,B,C,D,E}$ as a bipartite state
$|\psi\rangle_{\{A,B\},\{C,D,E\}}$.

The coarse-graining map becomes particularly useful when we consider LU-invariant functions that are symmetric under
permutations of the parties. Unequal local dimensions can appear to obstruct the definition of such functions, but this
is easily bypassed by embedding each local Hilbert space into a space of sufficiently high dimension, say $d$.  Choose linear isometric injections\footnote{Equivalently, fix an isometric identification of
$\mathcal H_{A_a}$ with a $d_a$-dimensional subspace of $\tilde{\mathcal H}_{A_a}\cong \mathbb C^{d}$.}
$i_a:\mathcal H_{A_a}\hookrightarrow \tilde{\mathcal H}_{A_a}$ with $\dim \tilde{\mathcal H}_{A_a}=d$ for all $a$, and
define the embedded state
\[
|\tilde\psi\rangle := \bigotimes_{a=1}^q i_a\,|\psi\rangle \;\in\; \bigotimes_{a=1}^q \tilde{\mathcal H}_{A_a}.
\]
Any LU-invariant function of $|\tilde\psi\rangle$ induces an LU-invariant function of $|\psi\rangle$ by composition with
the embedding $|\psi\rangle\mapsto|\tilde\psi\rangle$. In what follows we restrict attention to LU invariants obtained
in this way.

The advantage of working with $|\tilde\psi\rangle$ is that all parties now have the same local dimension, so we may
freely permute the tensor factors and evaluate the same LU invariant on the permuted state.  Accordingly, from now on
we work with the embedded state but suppress the tilde notation; that is, we assume without loss of generality that all
local Hilbert spaces $\mathcal H_{A_a}$ have the same dimension $d$.

\begin{definition}
An LU-invariant function $f:\bigotimes_{a=1}^q \mathcal H_{A_a}\to\mathbb C$ is called \emph{symmetric} if it is
invariant under permutations of the parties.
\end{definition}
\begin{definition}
    Given a symmetric $q$-partite LU-invariant $f$, we define its restriction to $m$-partite states $f|_m$ as follows.
\begin{align}
    f|_m\bigl(|\psi\rangle\bigr)
    \;:=\;
    f\bigl(|\psi\rangle \otimes |0\rangle^{\otimes (q-m)}\bigr),
\end{align}
where $|\psi\rangle$ is an $m$-partite state and the argument of $f$ on the right-hand side is the $q$-partite state
obtained by adjoining $q-m$ additional parties in a fixed reference state $|0\rangle$. 
\end{definition}
\noindent 
As $f$ is LU-invariant, the restriction $f|_m$ does not depend on the choice of ancillary reference state.
\begin{definition}
Given a symmetric $m$-partite LU-invariant $f$, we define its  $\pi$-extension $f_\pi$ ($\pi \in \Pi_X, |\pi|=m$) on $q$-partite states as the LU-invariant obtained by composing $f$ with the $\pi$-coarse-graining map
${\rm grp}_\pi$ i.e.
\begin{align}\label{extension}
    f_{\pi}\bigl(|\psi\rangle\bigr)
    \;:=\;
    f\!\left({\rm grp}_\pi\bigl(|\psi\rangle\bigr)\right).
\end{align}
\end{definition}
The two notions above, restriction and $\pi$-extension, allow us to produce a class of LU-invariants labeled by $\pi\in \Pi_X$ starting from any symmetric LU-invariant by first restricting it to an $m$-partite state and then $\pi$-extending the resulting $m$-partite LU-invariant to $q$-partite states by composition with the coarse-graining map ${\rm grp}_\pi$.

\section{General construction}\label{general_construct}
In this section we will give a general construction of $q$-partite symmetric signals and pre-signals from a family of $m$-partite LU-invariants, $m\leq q$, satisfying a certain compatibility condition.  This requires us to  review some facts about the lattice structure of the partition set $\Pi_X$. A natural partial ordering exists on the set of partitions. If $\kappa$ is obtained from $\pi$ by further partitioning $\pi$ then we say that $\kappa$ is finer than $\pi$ or equivalently, $\pi$ is coarser than $\kappa$. We denote this as $\kappa\leq \pi$.
Given two elements $\pi$ and $\kappa$ of a partially ordered set (poset), we define their join $\pi\vee\kappa$ as the smallest element that is larger than or equal to both and their meet $\pi\wedge\kappa$ as the largest element that is smaller than or equal to both. If the poset has the property that the meet and join of any pair of elements is unique then it is called a lattice. The poset of partitions is a lattice. Example, if $\pi=\{\{A,B\},\{C\}\}$ and $\kappa=\{\{A,C\},\{B\}\}$ then $\pi\vee \kappa=\{\{A,B,C\}\}$ and $\pi\wedge \kappa=\{\{A\},\{B\},\{C\}\}$. Note that, for any $\pi$,
\begin{align}
    \pi\vee {\bf 0}=\pi,\quad \pi\vee {\bf 1}={\bf 1},\quad \pi\wedge {\bf 0}={\bf 0}, \quad \pi\wedge {\bf 1}=\pi.\notag
\end{align}

Lastly, we need to review the M\"obius inversion formula on the partition lattice \cite{Rota1964Mobius, Stanley2011EC1}. 
The M\"obius function $\mu(\kappa, \pi)$ is defined for a pair of  partitions $\kappa\leq \pi$. It is the unique function that obeys,
\begin{align}
    \sum_{\kappa \leq \tau\leq \pi} \mu(\kappa, \tau)=\delta_{\kappa, \pi}.
\end{align} 
As a result, it can be used to perform the following inversion.
If $f,g:\Pi_X\to A$ are functions of partitions into an abelian group $A$ such that
\begin{align}
    g(\pi)=\sum_{\kappa\leq \pi} f(\kappa)
\end{align}
then
\begin{align}
    f(\pi)=\sum_{\kappa\leq \pi} \mu(\kappa, \pi) g(\kappa).
\end{align}
This also holds if $\leq$ is replaced by $\geq$ in both of the above equations. M\"obius inversion on the partition lattice plays a role analogous to the passage from moments to cumulants in probability theory; see, for example, \cite{Speed1983Cumulants}.

The M\"obius function $\mu(\kappa, \pi)$ is evaluated as follows. As $\kappa \leq \pi$, blocks of $\kappa$ are obtained by decomposing blocks of $\pi$. Let $k_B$ be the number of $\kappa$ blocks obtained from the block $B$ of $\pi$. Then,
\begin{align}
    \mu(\kappa, \pi)=\prod_{B\in \pi} (-1)^{k_B-1}(k_B-1)!
\end{align}
Here the product is taken over all blocks $B$ of $\pi$. The special case of $\mu(\pi, {\bf 1})$ will be particularly useful to us. The general formula reduces to,
\begin{align}
    \mu(\kappa, {\bf 1})=(-1)^{|\kappa|-1}(|\kappa|-1)!
\end{align}

Now we are ready to get back to the general construction of additive signals and not-necessarily additive pre-signals.
\begin{definition}
    A family $\{f_{m}\}, m=1, \ldots, q$ of $m$-partite symmetric LU-invariants is called compatible if it obeys
    \begin{align}\label{compatible}
        f_{|\pi|,\pi}(|\psi\rangle_\kappa)=f_{|\pi\wedge \kappa|,\pi\wedge \kappa}(|\psi\rangle_\kappa).
    \end{align}
    Here $|\psi\rangle_\kappa$ is a $\kappa$-separable state.
\end{definition}
\noindent
To avoid redundancy and clutter, we denote the $\pi$-extension of $f_m$ of a compatible family as $f_\pi$ rather than $f_{|\pi|, \pi}$.
\begin{remark}\label{extend}
    If $\{f_{m}^{(i)}\}, m=1, \ldots, q$  is a compatible family for each $i=1,\ldots, r$  then $\{g(f_{m}^{(i)})\}, m=1, \ldots, q$ is also a compatible family where $g:{\mathbb C}^r\to {\mathbb C}$.
\end{remark}

\begin{lemma}\label{additive_comp}
    If $f$ is a symmetric and additive $q$-partite LU-invariant then $\{f|_m\}, m=1,\ldots, q$ is a compatible family.
\end{lemma}
\noindent The lemma is proved in Appendix~\ref{lemma_proof}. We offer some intuition here through the following simple example. 

Let $X=\{A,B,C\}$, $\kappa=AB|C$, $\pi=A|BC$, so that
$\pi\wedge\kappa=A|B|C$.  Take the $\kappa$-separable state
$|\psi\rangle=|\psi_{AB}\rangle\otimes|\psi_C\rangle$.  We verify
$f_\pi(|\psi\rangle)=f_{\pi\wedge\kappa}(|\psi\rangle)$.

Coarse-graining by $\pi$ groups the parties as $A$ and $BC$, giving the
bipartite state $|\psi_{AB}\rangle\otimes|\psi_C\rangle$ on
$\mathcal{H}_A\otimes(\mathcal{H}_B\otimes\mathcal{H}_C)$.
Coarse-graining by $\pi\wedge\kappa$ instead treats all three parties
separately.  In both cases, additivity of $f$ reduces the computation to
the same two independent pieces: $f$ evaluated on $|\psi_{AB}\rangle$
and $f$ evaluated on $|\psi_C\rangle$.

\medskip\noindent\textit{Computing $f_\pi(|\psi\rangle_\kappa)$.}
Since $\pi=A|BC$, the map $\mathrm{grp}_\pi$ views
$|\psi\rangle_\kappa=|\psi_{AB}\rangle\otimes|\psi_C\rangle$ as the bipartite
state $|\psi_{AB}\rangle\otimes|\psi_C\rangle\in\mathcal{H}_A\otimes
(\mathcal{H}_B\otimes\mathcal{H}_C)$.  By the layer decomposition of Step~1,
this bipartite state splits into layers $|\psi\rangle^{(1)}_\pi=|\psi_{AB}
\rangle$ (occupying both parties nontrivially, $r_1=2=|\pi|$) and
$|\psi\rangle^{(2)}_\pi=|0\rangle_{C_1}\otimes|\psi_C\rangle_{C_2}$ (with
party $C_1$ in the reference state, since $C_1\cap B_2=\emptyset$).
Additivity gives $f|_2(|\psi\rangle^{(1)}_\pi)=f|_2(|\psi_{AB}\rangle)$,
and stripping the reference-state party via symmetry and Definition~5 gives
$f|_2(|\psi\rangle^{(2)}_\pi)=f|_1(|\psi_C\rangle)$, so
\[
f_\pi(|\psi\rangle_\kappa)
=f|_2\!\left(|\psi_{AB}\rangle\right)+f|_1\!\left(|\psi_C\rangle\right).
\]

\medskip\noindent\textit{Computing $f_{\pi\wedge\kappa}(|\psi\rangle_\kappa)$.}
Since $\pi\wedge\kappa=A|B|C$, the map $\mathrm{grp}_{\pi\wedge\kappa}$ is the
identity and we work with the tripartite state directly.  The two layers are
$|\psi\rangle^{(1)}_{\pi\wedge\kappa}=|\psi_{AB}\rangle\otimes|0\rangle_C$
and $|\psi\rangle^{(2)}_{\pi\wedge\kappa}=|0\rangle_A\otimes|0\rangle_B
\otimes|\psi_C\rangle$.  Stripping the reference-state parties sequentially
via Definition~5,
\begin{align}
f|_3\!\left(|\psi_{AB}\rangle\otimes|0\rangle_C\right)&=f|_2\!\left(|\psi_{AB}
\rangle\right),\notag\\
f|_3\!\left(|0\rangle_A\otimes|0\rangle_B\otimes|\psi_C\rangle\right)
&=f|_1\!\left(|\psi_C\rangle\right),
\end{align}
giving $f_{\pi\wedge\kappa}(|\psi\rangle_\kappa)=f|_2(|\psi_{AB}\rangle)
+f|_1(|\psi_C\rangle)$.

\medskip\noindent\textit{Why the two agree.}
Despite operating at different coarse-graining levels -- bipartite for
$f_\pi$ and tripartite for $f_{\pi\wedge\kappa}$ -- both computations reduce
via additivity and stripping to the same two contributions: $f|_2$ on
$|\psi_{AB}\rangle$ and $f|_1$ on $|\psi_C\rangle$.  The separability of
$|\psi\rangle_\kappa$ across $\kappa$ is what forces this reduction, making
the intermediate coarse-graining level irrelevant.

Lemma~\ref{additive_comp} and Remark~\ref{extend} yield an infinite class of compatible families. Lemma~\ref{additive_comp} gives us a compatible family that is additive. If we choose $g$ to be a non-linear function then from additive compatible families we get a non-additive compatible family. 

Now we will give general constructions of the following:
\begin{itemize}
    \item Symmetric additive signal from an additive compatible family.
    \item Symmetric non-additive pre-signal from a non-additive compatible family. 
\end{itemize}

\subsection{Symmetric signal and pre-signal}
We will make the following ansatz for a signal:
\begin{align}\label{ansatz_signal}
    s(|\psi\rangle)\;:=\;\sum_{\pi\in \Pi_X} c_\pi\, f_\pi(|\psi\rangle).
\end{align}
Here $f_\pi:=f_{|\pi|,\pi}$ denotes the extension of a $|\pi|$-partite LU-invariant from an \emph{additive compatible}
family to $q$-partite states via coarse-graining. Since $s$ is a linear combination of additive LU-invariants, it is
additive as well. By Lemma~\ref{key_lemma}, it suffices to impose the vanishing condition on $\kappa$-separable states
with type $\lambda(\kappa)=(1,q-1)$. Using compatibility~\eqref{compatible}, we obtain
\begin{align}\label{signal_constraints}
    &\sum_{\pi\in\Pi_X} c_\pi\, f_{\pi\wedge \kappa}\bigl(|\psi\rangle_\kappa\bigr)=0,\notag\\
    &\quad\text{for all }\kappa\text{ with }\lambda(\kappa)=(1,q-1).
\end{align}

\medskip
We use the same strategy to obtain a symmetric (generally non-additive) pre-signal from a non-additive compatible
family. We make the analogous ansatz
\begin{align}\label{ansatz_presignal}
    p(|\psi\rangle)\;:=\;\sum_{\pi\in \Pi_X} c_\pi\, \tilde f_\pi(|\psi\rangle).
\end{align}
Here $\tilde f_\pi:=\tilde f_{|\pi|,\pi}$ denotes the extension of a $|\pi|$-partite LU-invariant from a (not
necessarily additive) compatible family to $q$-partite states. For $p(|\psi\rangle)$ to be a pre-signal, it must
vanish on all separable states, i.e.\ on all $\kappa$-separable states with $\kappa\neq \mathbf 1$.
Using compatibility~\eqref{compatible}, this becomes
\begin{align}\label{presignal_constraints}
    \sum_{\pi\in\Pi_X} c_\pi\, \tilde f_{\pi\wedge \kappa}\bigl(|\psi\rangle_\kappa\bigr)=0,
    \qquad \text{for all }\kappa\neq \mathbf 1 .
\end{align}

\medskip
The structural similarity of \eqref{signal_constraints} and \eqref{presignal_constraints} becomes transparent once we
encode the coefficients $\{c_\pi\}$ as a formal linear combination of partitions in $\mathbb C[\Pi_X]$, with the meet
operation extended bilinearly:
\begin{align}\label{formal_vector}
    V\;:=\;\sum_{\pi\in\Pi_X} c_\pi\,\pi\ \in\ \mathbb C[\Pi_X].
\end{align}
Given a subset ${\cal K}\subseteq \Pi_X$, we impose the \emph{meet-vanishing constraints}
\begin{align}\label{meet_vanishing}
    V\wedge \kappa \;=\; 0,\qquad \forall\,\kappa\in{\cal K}.
\end{align}
The two problems above correspond to different choices of ${\cal K}$ (singleton cuts in the additive case, and
${\cal K}=\Pi_X\setminus\{\mathbf 1\}$ in the pre-signal case). In Appendix~\ref{useful_lemmas} we solve
\eqref{meet_vanishing} for arbitrary ${\cal K}$. 

Let us denote by $\downarrow{\cal K}$ the downset of ${\cal K}$ namely the  closure of  ${\cal K}$ under refinements:
\begin{align}
    \downarrow{\cal K}=\{\rho\in \Pi_X: \exists \kappa\in {\cal K}\ \text{with}\ \rho\leq \kappa\}.
\end{align}
We show in Appendix~\ref{useful_lemmas} that,
\begin{align}
    V\wedge {\cal K}=0\quad \Rightarrow \quad V\wedge \downarrow{\cal K}=0.
\end{align}
Denoting $\Pi_X\setminus \downarrow {\cal K}=:{\cal N}$, we also show that the most general solution to \eqref{meet_vanishing} is
\begin{align}\label{general_solution_meet}
    V \;=\; \sum_{\rho\in{\cal N}} a_\rho\, M_\rho,
    \quad\text{where}\quad
    M_\rho \;:=\; \sum_{\pi\le \rho} \mu(\pi,\rho)\,\pi,
\end{align}
for arbitrary coefficients $a_\rho\in\mathbb C$, and where $\mu$ denotes the M\"obius function of the partition lattice. In the case of additive signals, ${\cal N}$ consists of partitions that do not contain any block of size $1$ and in the case of non-additive pre-signals, ${\cal N}$ consists of a single partition ${\bf 1}$. Hence we have two theorems,
\begin{theorem}
    If $f_m$ is an additive compatible family of LU-invariants then
    \begin{align}
        \sum_{\rho\in{\cal N}} a_\rho\, M_\rho[f],
    \quad\text{where}\quad
    M_\rho[f] \;:=\; \sum_{\pi\le \rho} \mu(\pi,\rho)\,f_\pi,
    \end{align}
    is an additive signal for $a_\rho\in {\mathbb C}$. Here ${\cal N}$ consists of all partitions that do not have a block of size $1$.
\end{theorem}
\noindent 
In Appendix~\ref{useful_lemmas}, we also show that, in fact the space of additive signals is spanned by $M_\rho[f]$ for $\rho \in {\cal N}$.
\begin{theorem}
    If $\tilde f_m$ is a compatible family of LU-invariants then
    \begin{align}
    M_{\bf 1}[\tilde f]=\sum_{\pi\neq {\bf 1}} \mu(\pi,{\bf 1})\,{\tilde f}_\pi,
    \end{align}
    is a  pre-signal.
\end{theorem}
\noindent
In Appendix~\ref{useful_lemmas}, we also show that $M_{\bf 1}[\tilde f]$ is a unique pre-signal (up to scale) for a generic (non-additive) compatible family ${\tilde f}_m$.

\subsection{Examples}
In this subsection, we give explicit examples of signals constructed from some simple compatible families.
We construct these compatible families by choosing a symmetric $q$-partite additive LU-invariant as a seed and taking
its $m$-partite restrictions.  We will show that almost all signals  that appear
in the literature arise as special cases of our construction. The titles of the following subsections indicate the
choice of seed LU-invariant $f$.

\subsubsection{R\'enyi entropy}\label{renyi}
Let the seed $q$-partite LU-invariant be
\[
f_1=\sum_{a=1}^q S^{(n)}_{A_a},
\qquad
S^{(n)}_{A_a}=\frac{1}{1-n}\log \mathrm{Tr}\,\rho_{A_a}^n,
\]
the sum of single-party $n$th R\'enyi entropies. This is symmetric and additive, hence its restrictions $f|_m$
form an additive compatible family. The vector space of additive signals is spanned by \emph{M\"obius vectors}
$M_\rho[f_1]$, where $\rho$ has no singleton blocks. Interestingly, for odd $q$ all such $M_\rho[f_1]$ vanish, while for even
$q$ all $M_\rho[f_1]$ vanish except for $\rho=\mathbf 1$. We prove this in Appendix~\ref{vanishing}.
Below we evaluate $M_{\mathbf 1}[f_1]$ explicitly for $q=2,4$ and $6$, using purity of $|\psi\rangle$ to simplify. We use the symbols $A,B,C,\ldots$ for parties.

\begin{itemize}
\item $q=2$: $s=S^{(n)}_{AB}-S^{(n)}_{A}-S^{(n)}_{B}=-2S^{(n)}_{A}$.
\item $q=4$:$
s=S^{(n)}_{ABCD}-(S^{(n)}_{ABC}+\ldots)+(S^{(n)}_{AB}+\ldots)-(S^{(n)}_{A}+\ldots)
=(S^{(n)}_{AB}+\ldots)-2(S^{(n)}_{A}+\ldots).
$
\end{itemize}
For $q=6$, the purity-simplified expression is
\begin{itemize}
\item $q=6$: $s=(S^{(n)}_{ABC}+\ldots)-2(S^{(n)}_{AB}+\ldots)+2(S^{(n)}_{A}+\ldots)$.
\end{itemize}
This is precisely the $q$-information discussed in \cite{Balasubramanian:2024ysu}. As shown in Appendix~\ref{vanishing}, for general $q$, the signal takes the form
\begin{align}
    M_{\bf 1}[f_1]= \sum_{A\subset X} (-1)^{q-|A|}\,\,S^{(n)}_A.
\end{align}

\subsubsection{Sum of R\'enyi entropy and reflected entropy}
The signals constructed in Section~\ref{renyi} vanish for odd $q$.
To mitigate this issue, the authors of \cite{Balasubramanian:2024ysu} introduced a new signal called the
\emph{residual information}. For even $q$ it reproduces the $q$-information, while for odd $q$ it is non-vanishing for GME states.

Given a $q$-partite pure state with odd $q$, trace out one party, say $A_q$, obtaining the density matrix
$\rho_{A_1,\ldots,A_{q-1}}$ on $q-1$ parties. Canonically purify it to
$|\psi\rangle_{A_1,A_1^*,\ldots,A_{q-1},A_{q-1}^*}$. The seed LU-invariant is then
\[
f_2=\sum_{a=1}^{q-1}\left(\frac12 S_{A_aA_a^*}-S_{A_a}\right).
\]
If $A_q$ factorizes (so that $\rho_{A_1,\ldots,A_{q-1}}$ is pure), this seed invariant vanishes.
The corresponding M\"obius vector $M_{\mathbf 1}[f_2]$ vanishes whenever the $(q-1)$-partite density matrix factorizes.
This $M_{\mathbf 1}[f_2]$ is called the $q$-residual information.

\subsubsection{R\'enyi multi-entropy}
Let the seed LU-invariant to be the $n$th R\'enyi multi-entropy $f_3=S^{(n)}_{ q}$. For $n=1$ this reduces to the
multi-entropy. The R\'enyi multi-entropy is proportional to the logarithm of a multi-invariant; since multi-invariants
are multiplicative, R\'enyi multi-entropy is additive. Our construction reproduces the so-called genuine multi-entropy
$\mathrm{GM}^{(n)}$ introduced in \cite{Iizuka:2025ioc, Iizuka:2025caq}. An advantage of our approach is that it provides a
free parametrization of the signal space in terms of the spanning M\"obius vectors $M_\rho[f_3]$.

Below we evaluate the relevant spanning vectors for $q=2,3,4$. For $q=2,3$, the set $\mathcal N$ contains only
$\mathbf 1$:
\begin{itemize}
\item $q=2$: $s=-S^{(n)}_2(A,B)$.
\item $q=3$: $s=-(S^{(n)}_2(AB,C)+\ldots) +2\,S^{(n)}_3(A,B,C).$
\end{itemize}
Here we used $S^{(n)}_1(AB)=0$ for $q=2$ and $S^{(n)}_1(ABC)=0$ for $q=3$.

For $q=4$, the set $\mathcal N$ consists of the three pair partitions $\{AB|CD\}$, $\{AC|BD\}$, $\{AD|BC\}$, together
with $\mathbf 1$. Using $S^{(n)}_1(ABCD)=0$, we obtain
\begin{itemize}
\item $M_{AB|CD}= - S^{(n)}_3(A,B,CD) - S^{(n)}_3(AB,C,D) + S^{(n)}_4(A,B,C,D)$.
\item $M_{ABCD}=-(S^{(n)}_2(ABC,D)+\ldots)\\\\ -(S^{(n)}_2(AB,CD)+\ldots)
+2(S^{(n)}_3(AB,C,D)+\ldots)-6S^{(n)}_4(A,B,C,D)$.
\end{itemize}
Thus $\mathrm{GM}^{(n)}$ for $q=4$ is spanned by $M_{AB|CD}$, $M_{AC|BD}$, $M_{AD|BC}$ and $M_{ABCD}$.
For higher $q$, the spanning M\"obius vectors follow straightforwardly from the explicit formula for the M\"obius
function on $\Pi_X$.

\subsubsection{Multipartite entanglement of purification}
Finally, take the seed LU-invariant to be the multipartite entanglement of purification \cite{Umemoto:2018jpc},
\[
f_4=E_p^{(q)}:=\min\Big(\sum_{a=1}^{ q} S_{A_a\tilde A_a}\Big),
\]
defined for a $q$-partite mixed state $\rho_{A_1,\ldots,A_{q}}$. One purifies $\rho$ to
$|\psi\rangle_{A_1,\ldots,A_{q},\tilde A_1,\ldots,\tilde A_{q}}$ and minimizes over all purifications.
For pure states, $E_p^{(q)}=\sum_{a=1}^q S_{A_a}$ and hence it reduces to the first example with $n=1$. For mixed states, a
linear combination of $E_p^{(q)}$ can be formed that vanishes whenever the density matrix is layerwise-factorized. For
mixed states, the term separable is often reserved for density matrices that are convex combinations of product
states. The signal constructed here vanishes only on strictly factorized density matrices.

Since $E_p^{(1)}=0$, the explicit form of the resulting signal mimics that of the R\'enyi multi-entropy example above.
Below we record only the coarsest M\"obius vector $M_{\mathbf 1}[f_4]$ for $q=3$ and $q=4$:
\begin{itemize}
\item $q=3$: $s=-(E^{(2)}_p(AB,C)+\ldots) +2\,E^{(3)}_p(A,B,C).$
\item $q=4$: $s=-(E^{(2)}_p(ABC,D)+\ldots) \\\\-(E^{(2)}_p(AB,CD)+\ldots)
+2(E^{(3)}_p(AB,C,D)+\ldots)-6E^{(4)}_p(A,B,C,D)$.
\end{itemize}
Up to an overall sign, these quantities coincide with the multipartite correlation signals $\Delta_p^{(q)}$
introduced in \cite{Bao:2025psl}. Note that in general the signal space is spanned by $M_\rho[f_4]$ with $\rho\in\mathcal N$
and not only by $M_{\mathbf 1}[f_4]$.

\section{Non-symmetric Signals from multi-invariants}
In this subsection, we will construct a class of non-symmetric additive signals from a general $q$-partite multi-invariant. 
Let us start with the definition of multi-invariant.
Let us write the $q$-partite state $|\psi\rangle$ in an arbitrarily chosen but factorized basis
\begin{align}
    |\psi\rangle=\sum_{i_1, \ldots, i_q} \psi_{i_1, \ldots, i_q} |i_1\rangle\otimes \ldots \otimes |i_q\rangle
\end{align}
Here $|i_a\rangle$ labels the basis of the Hilbert space factor ${\cal H}_a$. Multi-invariants are homogeneous polynomials of $\psi_{i_1, \ldots, i_q}$ and its complex conjugate that are multiplicative \cite{Gadde:2024taa}.
Local unitary invariant functions of a
multi-partite quantum states have been explored widely.
See \cite{Vrana_2011, Vrana_2012, Szalay_2012, hero2011measure, HERO20096508, Rains2000PolynomialInvariantsQuantumCodes, BrylinskiBrylinski2002InvariantPolyQuditsChapter, GrasslRoettelerBeth1998ComputingLocalInvariants} for their enumeration and construction. Multi-invariants are constructed by taking a fixed number, say $n$, of $\psi$'s and the same number of $\bar \psi$'s and contracting all the indices. 
The index contraction pattern can be described by an edge-colored graph. See \cite{Gadde:2022cqi, Gadde:2023zni, Gadde:2024jfi, Gadde:2025ybn} for some examples of such graphs. It is also specified by an $n$-tuple $(\sigma_1,\ldots, \sigma_q)$ of permutation elements of $S_n$. A multi-invariant can then be compactly written as
\begin{align}\label{Z_def}
    {\cal Z}(\sigma_1, \ldots, \sigma_q;|\psi\rangle):=\langle \psi|^{\otimes n} \,\sigma_1\otimes \ldots\otimes \sigma_q\,|\psi\rangle^{\otimes n}.
\end{align}
We have introduced the subscript of permutation elements to indicate the type of multi-invariant. From this definition, it is clear that it is multiplicative under layer decomposition. 
The labeling by permutation tuple contains a redundancy. It is due to the freedom in relabeling bra's and ket's. It gives the relation,
\begin{align}
    {\cal Z}(\sigma_1, \ldots, \sigma_q;\cdot)&={\cal Z}(g\cdot\sigma_1, \ldots, g\cdot\sigma_q;\cdot)\notag\\
    &={\cal Z}(\sigma_1\cdot h, \ldots, \sigma_q\cdot h;\cdot), \qquad g,h\in S_n\notag
\end{align}
A general multi-invariant need not be positive or even real, but using Cauchy-Schwarz inequality, we see from equation \eqref{Z_def}  that $|{\cal Z}(\cdot;\cdot)| \leq 1$.  It is however possible to construct multi-invariants that are real and positive for all states. Its edge-colored graph has the property that it is reflection symmetric which allows one to interpret the multi-invariant as a norm of a composite tensor. 

For positive multi-invariants, we define the logarithmic multi-invariants ${\cal E}(\cdot,\cdot)\equiv -\frac{1}{n}\log {\cal Z} (\cdot,\cdot)$.
A straightforward consequence of the multiplicativity of ${\cal Z}$'s is the additivity of ${\cal E}$.
We use this additivity in conjunction with Lemma~\ref{key_lemma}  to construct multi-partite signals.
\begin{theorem}
    If $T$ is an $s$-dimensional rank $q$ tensor whose entries are in ${\mathbb C}$ with the property $\sum_{i_k}T_{i_1, \ldots, i_q}=0,\forall k$ with other indices kept fixed then the following is a signal
    \begin{align}\label{signal_construct}
        \sum_{i_i,\ldots, i_q} T_{i_1, \ldots, i_q}{\cal E}(\sigma_{i_1}, \ldots, \sigma_{i_q})
    \end{align}
    Here  $\sigma_i\in \Sigma\subseteq S_n$ and $|\Sigma|=s$.
\end{theorem}
\begin{proof}
If party $A_1$ factorizes then any $q$-partite logarithmic multi-invariant obeys,
\begin{align}\label{vanishing1}
    {\cal E}(\sigma_i, \sigma_2, \ldots;\cdot)={\cal E}(\sigma_j, \sigma_2, \ldots;\cdot) \qquad \forall \sigma_i,\sigma_j\in S_n.
\end{align}
To make full use of this condition, we promote permutation elements in the argument of ${\cal E}$ to a formal linear combination permutation elements and define
\begin{align}
    {\cal E}(\sum_j \alpha_j\sigma_j,\ldots;\cdot )=\sum_j \alpha_j{\cal E}(\sigma_j,\ldots;\cdot ).
\end{align}
This definition is also extended to other permutation arguments of ${\cal E}$. Then equation \eqref{vanishing1} can be generalized to
\begin{align}
    {\cal E}(\sum_j \alpha_j\sigma_j,\ldots;\cdot )=0,\qquad {\rm when}\quad \sum_j \alpha_j=0.
\end{align}
where the rest of the arguments of ${\cal E}$ can also be formal linear combinations of permutations. The vanishing of $f$ for states where any other party is factorized is solved in the same way. The following linear combination of logarithmic multi-invariants is a signal
\begin{align}
    {\cal E}(\sum_j \alpha_j^{(1)}\sigma_j,\ldots, \sum_j \alpha_j^{(q)}\sigma_j;\cdot),\qquad {\rm where} \sum_j \alpha_j^{(i)}=0, \forall i.\notag
\end{align}
Of course, linear combinations of such logarithmic multi-invariants are also signals:
\begin{align}
    \sum_{j_1,\ldots, j_q} (\alpha_{j_{1}}^{(1)}\ldots \alpha_{j_{q}}^{(q)}+\beta_{j_{1}}^{(1)}\ldots \beta_{j_{q}}^{(q)}+\ldots)\,\, {\cal E}(\sigma_{j_1},\ldots \sigma_{j_q};\cdot)\notag
\end{align}
where $\alpha's, \beta's, \ldots$ etc obeys $\sum_j \alpha_j^{(i)}=0$ for all $j$. The coefficient of ${\cal E}$ is precisely the tensor $T$ appearing in the statement of the theorem.
\end{proof}
\noindent
As we can see, there is a large ambiguity in the signals constructed from a general positive multi-invariant. For a general positive multi-invariant, a signal that contains the least number of terms can be chosen to be
\begin{align}
    {\cal E}(\sigma_1-\sigma_2,\sigma_2-\sigma_1, \sigma_3-\sigma_1,\ldots, \sigma_q-\sigma_1;\cdot).
\end{align}

\section{Conclusion and outlook}

We introduced a general route to genuinely multipartite entanglement signals from symmetric families of lower-partite local-unitary invariants. The central observation is that the partition lattice provides the natural bookkeeping device for multipartite separability, and that M\"obius inversion on this lattice isolates combinations that vanish on all nontrivially separable states. This yields a simple universal mechanism for turning compatible families of local-unitary invariant quantities into genuinely multipartite signals. Such compatible families can in turn be constructed naturally from additive local-unitary invariants. We also showed how to extract a genuinely multipartite entanglement signal from multi-invariants that are not necessarily symmetric. Signals of this kind have recently been applied to ground-state wavefunctions of gapped quantum systems in the study of topological phases \cite{DelZotto:2026fpw}.

Several directions remain open. It would be interesting to explore if the pre-signals constructed here can be used to make GME measures. This would amount to searching for pre-signals that are non-negative and monotonically non-decreasing under LOCC. We expect that the graph based technology, in particular the notion of edge-convexity discovered in \cite{Gadde:2024jfi,Gadde:2024taa} to be useful in this search. The GME measures would be particularly useful if they can be computed for mixed states. If pre-signals are constructed from multi-invariants, which are polynomials in the state coefficients, then the extension to mixed states via convex roof may be more tractable. 

The condition of compatibility is natural, as is evident from the fact that such families can be constructed from additive LU-invariants, but its explicit form is cumbersome. It would be nice to understand this condition structurally to identify all LU-invariant families that are compatible. 

Lastly, Ref.~\cite{Ju:2026xfh} proposes a hypergraph-based organization of multipartite entanglement structures and derives linear relations among partition-labeled measures that detect entanglement patterns. In that framework, local Hilbert spaces are refined into tensor factors, resource states are placed on hyperedges. From our viewpoint this defines a subclass of layer-decomposable states. On the other hand, we work directly with the full partition lattice and compatible families of partition-sensitive local-unitary invariants, and use Möbius inversion to construct genuinely multipartite signals. It would be interesting to relate the signal spaces of the two papers and to understand when the two constructions coincide.

\begin{acknowledgments}
We thank Sriram Akella, Shraiyance Jain and Pratik Rath for useful discussions. This work is supported by the Department of Atomic Energy, Government of India, under Project Identification No. RTI 4002, and the Infosys Endowment for the study of the Quantum Structure of Spacetime.
\end{acknowledgments}

\appendix

\section{Proof of Theorem 1}\label{no_signal}
\noindent {\bf Theorem 1}. \emph{No signal can serve as a genuine multipartite entanglement (GME) measure.}
\begin{proof}

    Consider the following state that is layerwise-separable.
    \begin{align}
        |\psi\rangle=|\psi^{(1)}\rangle\otimes |\psi^{(2)}\rangle
    \end{align}
    where $|\psi^{(1)}\rangle=|{\rm Bell}\rangle_{A_1B_1}\otimes |0\rangle_{C_1}$ and $|\psi^{(2)}\rangle=|0\rangle_{A_2}\otimes |{\rm Bell}\rangle_{B_2 C_2}$. 
    Here we have used the labels $A,B$ and $C$ for the three parties. Each of the parties is further factorized into $A=A_1\otimes A_2$ and so on. The Bell state is defined as
    \begin{align}
        |{\rm Bell}\rangle_{A_1B_1}=\frac{1}{\sqrt 2}(|00\rangle_{A_1B_1}+|11\rangle_{A_1B_1}).
    \end{align}
    Consider the following local operation on party $B$,
    \begin{align}
        E_1&=(|00\rangle\langle00|+|11\rangle\langle11|)_{B_1 B_2}\notag\\
        E_2&=(|01\rangle\langle01|+|10\rangle\langle10|)_{B_1 B_2}
    \end{align}
    It is easy to check that the set $\{E_1,E_2\}$ is trace preserving. After application of these Krauss operators, we get
    \begin{align}
        |\psi\rangle_1&=\frac{1}{\sqrt 2}(|0\rangle_{A_1}|00\rangle_{B_1B_2}|0\rangle_{C_1}+|1\rangle_{A_1}|11\rangle_{B_1B_2}|1\rangle_{C_1})\notag\\
        |\psi\rangle_2&=\frac{1}{\sqrt 2}(|0\rangle_{A_1}|01\rangle_{B_1B_2}|1\rangle_{C_1}+|1\rangle_{A_1}|10\rangle_{B_1B_2}|0\rangle_{C_1}).\notag
    \end{align}
    Neither of these is a layerwise-separable state. In fact both of them are LU-equivalent to a GHZ state. It is manifest for $|\psi\rangle_1$, for $|\psi_2\rangle$ map $|01\rangle_{B_1B_2}\to |0\rangle_{B}$ and $|1\rangle_C\to|0\rangle_C$.
    We expect that a GME measure  be strictly positive on both of them. This is in conflict with the property that a GME measure should be monotonically non-increasing under LOCC on average. 
\end{proof}

\section{Proof of Lemma 2}\label{lemma_proof}
\noindent {\bf Lemma 2}. If $f$ is a symmetric and additive $q$-partite LU-invariant then $\{f|_m\}, m=1,\ldots, q$ is a compatible family.
\begin{proof}
Denote $f_\pi := f|_{|\pi|,\pi}$.  We show that for every $\kappa$-separable state
\[
|\psi\rangle_\kappa
  = \bigotimes_{j=1}^{m}|\psi_{B_j}\rangle,
\qquad
|\psi_{B_j}\rangle\in\bigotimes_{A_a\in B_j}\mathcal{H}_{A_a},
\]
we have $f_\pi(|\psi\rangle_\kappa)=f_{\pi\wedge\kappa}(|\psi\rangle_\kappa)$.

\medskip\noindent\textit{Step 1: Layer decomposition of $\mathrm{grp}_\pi(|\psi\rangle_\kappa)$.}
For each block $C_i\in\pi$, factor the coarse-grained Hilbert space as
\[
\mathcal{H}_{C_i}
  = \bigotimes_{A_a\in C_i}H_{A_a}
  = \bigotimes_{j=1}^{m}\mathcal{H}_{C_i\cap B_j}^{(j)},
\]
where $\mathcal{H}_{C_i\cap B_j}^{(j)}:=\bigotimes_{A_a\in C_i\cap B_j}H_{A_a}$, equal to $\mathbb{C}$
(with state $|0\rangle$) when $C_i\cap B_j=\emptyset$.  Reordering the tensor factors gives the identification
\[
\bigotimes_{i}\mathcal{H}_{C_i}
  = \bigotimes_{j}\bigotimes_{i}\mathcal{H}_{C_i\cap B_j}^{(j)}
  = \bigotimes_{j}\mathcal{H}_{B_j},
\]
where the last equality uses $\bigsqcup_i(C_i\cap B_j)=B_j$.  Under this identification
$|\psi\rangle_\kappa=\bigotimes_j|\psi_{B_j}\rangle$ is already written as a tensor product of $m$ layers, so $\mathrm{grp}_\pi(|\psi\rangle_\kappa)$ admits the layer decomposition
\[
\mathrm{grp}_\pi(|\psi\rangle_\kappa)
  = \bigotimes_{j=1}^{m}|\psi\rangle_\pi^{(j)},
\]
where the $j$-th layer $|\psi\rangle_\pi^{(j)}$ is the $|\pi|$-partite state on
$\bigotimes_i\mathcal{H}_{C_i\cap B_j}^{(j)}$
consisting of $|\psi_{B_j}\rangle$ on the $r_j$ parties with $C_i\cap B_j\neq\emptyset$
and $|0\rangle$ on the remaining $|\pi|-r_j$ parties (which have
$\mathcal{H}_{C_i\cap B_j}^{(j)}=\mathbb{C}$),
with $r_j:=|\{C_i\in\pi:C_i\cap B_j\neq\emptyset\}|$.

\medskip\noindent\textit{Step 2: Reducing each layer to $f|_{r_j}$.}
Since $f|_{|\pi|}$ is additive (as the restriction of an additive $f$), the layer
decomposition gives
\[
f_\pi(|\psi\rangle_\kappa)
  = f|_{|\pi|}\!\left(\mathrm{grp}_\pi(|\psi\rangle_\kappa)\right)
  = \sum_{j=1}^{m}f|_{|\pi|}\!\left(|\psi\rangle_\pi^{(j)}\right).
\]
For each $j$, the $|\pi|-r_j$ parties of $|\psi\rangle_\pi^{(j)}$ with
$\mathcal{H}_{C_i\cap B_j}^{(j)}=\mathbb{C}$ are in state $|0\rangle$.
After embedding into the common $d$-dimensional space (as in Definition~4),
these are separate parties of the $|\pi|$-partite state in state $|0\rangle\in\mathbb{C}^d$.
Using the symmetry of $f$ to reorder parties and then Definition~5,
\begin{align}
&f|_{|\pi|}\!\left(|\psi\rangle_\pi^{(j)}\right)\\
  &= f|_{r_j}\!\left(|\psi_{B_j}\rangle
    \text{ on the }r_j\text{ parties }
    \{C_i\cap B_j:C_i\cap B_j\neq\emptyset\}\right).\notag
\end{align}

\medskip\noindent\textit{Step 3: Identical reduction for $f_{\pi\wedge\kappa}$.}
The blocks of $\pi\wedge\kappa$ are $\{C_i\cap B_j:C_i\cap B_j\neq\emptyset\}$.
Since $|\psi\rangle_\kappa$ is $\kappa$-separable,
$\mathrm{grp}_{\pi\wedge\kappa}(|\psi\rangle_\kappa)$ is a product across the $m$
groups of parties $\{C_i\cap B_j\}_i$ (one group per block $B_j$).
The same argument—additivity of $f|_{|\pi\wedge\kappa|}$ followed by stripping the
reference-state parties $\{C_i\cap B_l:l\neq j\}$, which are separate parties
of the $|\pi\wedge\kappa|$-partite state in state $|0\rangle$, via symmetry and
Definition~5—yields
\[
f_{\pi\wedge\kappa}(|\psi\rangle_\kappa)
  = \sum_{j=1}^{m}f|_{r_j}\!\left(|\psi_{B_j}\rangle
    \text{ on }\{C_i\cap B_j:C_i\cap B_j\neq\emptyset\}\right).
\]
This coincides term-by-term with the sum in Step~2, so
$f_\pi(|\psi\rangle_\kappa)=f_{\pi\wedge\kappa}(|\psi\rangle_\kappa)$.
Hence $\{f|_m\}_{m=1}^q$ is a compatible family.
\end{proof}

\section{Useful lemmas}\label{useful_lemmas}
For the discussion in this section, it is convenient to let $\mathbb C[\Pi_X]$ be the free complex vector space with basis $\Pi_X$, and extend $\wedge, \vee$ bilinearly to $\mathbb C[\Pi_X]\times \mathbb C[\Pi_X]$.

\begin{remark}
    Let $F\in {\mathbb C}[\Pi_X]$ and $\kappa, \rho\in \Pi_X$ such that $\rho\leq \kappa$. If $F\wedge \kappa=0$, then $F\wedge \rho=0$.
\end{remark}
\begin{proof}
Using associativity of $\wedge$ in the lattice and bilinearity,
\[
(F\wedge \kappa)\wedge \lambda \;=\; F\wedge (\kappa\wedge \lambda).
\]
If $\lambda\le \kappa$, then $\kappa\wedge \lambda=\lambda$, hence
\[
(F\wedge \kappa)\wedge \lambda \;=\; F\wedge \lambda.
\]
If $F\wedge \kappa=0$, then the left-hand side equals $0\wedge\lambda=0$, so $F\wedge \lambda=0$.
\end{proof}
\begin{remark}\label{downset_closure}
    Let $\mathcal K\subseteq \Pi_X$ and define the solution space
\[
\mathcal V(\mathcal K)\;:=\;\{F\in\mathbb C[\Pi_X]:\ F\wedge \kappa=0\ \ \forall\,\kappa\in\mathcal K\}.
\]
Then
\[
\mathcal V(\mathcal K)\;=\;\mathcal V(\downarrow \mathcal K).
\]
The downset $\downarrow \mathcal K$ of $\cal K$ is the closure of $\cal K$ under refinements i.e. 
\[
\downarrow \mathcal K \;:=\;\{\lambda\in \Pi_X:\exists\,\kappa\in\mathcal K\text{ with }\lambda\le \kappa\}.
\]
\end{remark}
\noindent 
As a consequence of Remark~\ref{downset_closure}, without loss of generality we take ${\cal K}$ to be a downset in the following. 

\begin{lemma}
    The following two sets are equal:
    \begin{align}
        \{\pi \leq \rho, \pi \wedge \kappa= \tau\}=[\tau,\rho]\setminus \bigcup_{m\in M}[m, \rho]
    \end{align}
    where 
    \begin{align}
        M=\{m:\tau<m\leq  \rho\wedge \kappa, (\tau, m)=\varnothing\}
    \end{align}
\end{lemma}
\begin{proof}
We show both inclusions.

\smallskip

\noindent
``$\subseteq$'': Let $\pi\le \rho$ satisfy $\pi\wedge\kappa=\tau$. Since $\tau=\pi\wedge\kappa\le \pi$, we have $\tau\le \pi\le \rho$, hence $\pi\in[\tau,\rho]$. We claim that $\pi\notin [m,\rho]$ for any $m\in M$. Indeed, if $\pi\ge m$ for some $m\in M$, then since also $m\le \kappa$, we would have
\[
m\le \pi\wedge\kappa=\tau,
\]
contradicting $m>\tau$. Hence
\[
\pi\in [\tau,\rho]\setminus \bigcup_{m\in M}[m,\rho].
\]

\smallskip

\noindent
``$\supseteq$'': Let $\pi\in [\tau,\rho]\setminus \bigcup_{m\in M}[m,\rho]$. Since $\tau\le \pi$ and $\tau\le \kappa$, we have
\[
\tau\le \pi\wedge\kappa\le \rho\wedge\kappa.
\]
Suppose $\pi\wedge\kappa>\tau$. Then, because the interval $[\tau,\pi\wedge\kappa]$ is finite, there exists an element
\[
m\in [\tau,\pi\wedge\kappa]
\]
which is minimal subject to $m>\tau$ i.e. $(\tau,m)=\varnothing$, and since
\[
m\le \pi\wedge\kappa\le \rho\wedge\kappa,
\]
we have $m\in M$. Also $m\le \pi$, so $\pi\in [m,\rho]$, contradicting the assumption that
\[
\pi\notin \bigcup_{m\in M}[m,\rho].
\]
Therefore $\pi\wedge\kappa=\tau$.
\end{proof}
Let us get a feeling for this lemma by working out an example. 
Take $|X|=4$, let
\[
\rho = 1234,\qquad \kappa = 12|34,\qquad \tau = 1|2|3|4.
\]
Then
\[
\rho\wedge\kappa = 12|34,
\]
so the interval $[\tau,\rho\wedge\kappa]$ consists of
\[
1|2|3|4,\qquad 12|3|4,\qquad 1|2|34,\qquad 12|34.
\]
Hence the minimal elements above $\tau$ in this interval are
\[
M=\{\,12|3|4,\;1|2|34\,\}.
\]

The lemma says that
\[
\{\pi\le \rho:\ \pi\wedge\kappa=\tau\}
=
[\tau,\rho]\setminus\bigl([\,12|3|4,\rho]\cup[\,1|2|34,\rho]\bigr).
\]
Now $[\tau,\rho]$ is the whole partition lattice on $\{1,2,3,4\}$. Removing the partitions above $12|3|4$ excludes all partitions in which $1$ and $2$ lie in the same block; removing those above $1|2|34$ excludes all partitions in which $3$ and $4$ lie in the same block. Thus the surviving partitions are exactly
\[
1|2|3|4,\,\,\,
13|2|4,\,\,\,
14|2|3,\,\,\,
23|1|4,\,\,\,
24|1|3,\,\,\,
13|24,\,\,\,
14|23.
\]

We check directly that each of these has meet $\tau$ with $\kappa$. For instance,
\[
(13|24)\wedge(12|34)=1|2|3|4,
\]
since the block $13$ is split by $12|34$ into $\{1\}$ and $\{3\}$, and similarly $24$ is split into $\{2\}$ and $\{4\}$. On the other hand,
\[
(12|3|4)\wedge(12|34)=12|3|4>\tau,
\]
so $12|3|4$ is excluded, exactly as the lemma predicts.

\begin{lemma}\label{thm:kernel-meet-downset-correct}
Let $\mathcal K\subseteq \Pi_X$ be a down-set. Define
\begin{align}
    \mathcal V(\mathcal K)&:=\{F\in\mathbb C[\Pi_X]:\ F\wedge \kappa=0\ \ \forall\,\kappa\in\mathcal K\}\\
    M_\rho &:= \sum_{\pi\le \rho} \mu(\pi,\rho)\,\pi\ \in\ \mathbb C[\Pi_X].
\end{align}
Here $\mu(\pi, \rho)$ is the M\"obius function on the partition lattice $\Pi_X$.
Then 
\begin{enumerate}
    \item Vectors $\{M_\rho:\rho\in\Pi_X\}$ are linearly independent. 
    \item $\mathcal V(\mathcal K)\;=\;\mathrm{span}\{M_\rho:\rho\notin\mathcal K\}$
\end{enumerate}
\end{lemma}
\begin{proof}
\noindent
\text{(1) Linear independence.}
Suppose
\[
\sum_{\rho} a_\rho M_\rho=0.
\]
Since
\[
M_\rho=\rho+\sum_{\pi<\rho}\mu(\pi,\rho)\,\pi,
\]
the change-of-basis matrix from $\{M_\rho\}$ to the standard basis $\{\rho\}$ is unitriangular with respect to any linear extension of the refinement order. Therefore it is invertible, and the family $\{M_\rho\}$ is linearly independent.

\medskip

\noindent
\text{(2) Description of $\mathcal V(\mathcal K)$.}

For $\rho\in\Pi_X$ and $\kappa\in\Pi_X$, write
\[
M_\rho\wedge \kappa
=
\sum_{\pi\le \rho}\mu(\pi,\rho)\,(\pi\wedge \kappa).
\]
Fix $\tau\le \rho\wedge \kappa$. The coefficient of $\tau$ in $M_\rho\wedge\kappa$ is
\[
\sum_{\substack{\pi\le \rho\\ \pi\wedge\kappa=\tau}}\mu(\pi,\rho).
\]
By the previous lemma,
\[
\{\pi\le \rho:\ \pi\wedge\kappa=\tau\}
=
[\tau,\rho]\setminus \bigcup_{m\in M}[m,\rho],
\]
where
$
M=\{m:\tau<m\le \rho\wedge\kappa,\ (\tau,m)=\varnothing\}.
$
Applying inclusion--exclusion, this coefficient becomes
\[
\sum_{J\subseteq M}(-1)^{|J|}
\sum_{\pi\in [x_J,\rho]}\mu(\pi,\rho),
\qquad
x_J:=\bigvee(J\cup\{\tau\}).
\]
If $\rho\nleq \kappa$, then $\rho\wedge\kappa<\rho$, hence $x_J\le \rho\wedge\kappa<\rho$ for every $J$, so by the defining property of the M\"obius function,
\[
\sum_{\pi\in [x_J,\rho]}\mu(\pi,\rho)=0.
\]
Thus every coefficient vanishes, and therefore
\[
M_\rho\wedge\kappa=0
\qquad\text{whenever }\rho\nleq \kappa.
\]
If instead $\rho\le \kappa$, then for every $\pi\le \rho$ we have $\pi\wedge\kappa=\pi$, so
$
M_\rho\wedge\kappa=M_\rho.
$
Hence
\[
\boxed{
M_\rho\wedge\kappa=
\begin{cases}
M_\rho,& \rho\le \kappa,\\[4pt]
0,& \rho\nleq \kappa.
\end{cases}}
\tag{$*$}
\]
First, let $\rho\notin\mathcal K$. Since $\mathcal K$ is a down-set, there is no $\kappa\in\mathcal K$ with $\rho\le \kappa$, for otherwise $\rho\in\mathcal K$. Hence by $(*)$,
\[
M_\rho\wedge\kappa=0
\qquad\forall\,\kappa\in\mathcal K.
\]
So
\[
\mathrm{span}\{M_\rho:\rho\notin\mathcal K\}\subseteq \mathcal V(\mathcal K).
\]

Conversely, let
\[
F=\sum_{\rho} a_\rho M_\rho \in \mathcal V(\mathcal K),
\]
using part (1). Suppose $a_{\rho_0}\neq 0$ for some $\rho_0\in\mathcal K$. Choose such a $\rho_0$ maximal (with respect to refinement order) among all elements of $\mathcal K$ having nonzero coefficient. Since $\rho_0\in\mathcal K$ and $F\in\mathcal V(\mathcal K)$,
\[
0=F\wedge \rho_0=\sum_{\rho} a_\rho (M_\rho\wedge \rho_0).
\]
By $(*)$, only terms with $\rho\le \rho_0$ survive, so
\[
0=\sum_{\rho\le \rho_0} a_\rho M_\rho.
\]
Again by $(*)$, the coefficient of $M_{\rho_0}$ in this sum is exactly $a_{\rho_0}$, since no $\rho<\rho_0$ can produce $M_{\rho_0}$. Hence $a_{\rho_0}=0$, a contradiction. Therefore $a_\rho=0$ for every $\rho\in\mathcal K$, and so
\[
F\in \mathrm{span}\{M_\rho:\rho\notin\mathcal K\}.
\]

Thus
\[
\mathcal V(\mathcal K)=\mathrm{span}\{M_\rho:\rho\notin\mathcal K\}.
\qedhere
\]
\end{proof}

\section{Vanishing signals}\label{vanishing}
\begin{proposition}
    Let $f=\sum_{a=1}^q S^{(n)}_{A_a}$ be the seed LU-invariant of $q$-partite pure states. Let ${\cal N}$ be the set of partitions that do not contain a block of size $1$. Let
    \begin{align}
        M_\rho=\sum_{\pi\le \rho}\mu(\pi,\rho)\,f_\pi.
    \end{align}
    Then
    \begin{enumerate}
        \item $M_\rho$ for $\rho\in {\cal N}, \rho \neq {\bf 1}$ vanishes for all $q$.
        \item $M_{\bf 1}$ vanishes for odd $q$.
    \end{enumerate}
\end{proposition}
\begin{proof}
1. Write $\rho=\{R_1,\dots,R_m\}$ with $m=|\rho|\ge 2$ and $|R_i|\ge 2$ for all $i$.
Any refinement $\pi\le \rho$ is uniquely determined by partitions $\pi_i\in\Pi_{R_i}$ via
$\pi=\pi_1\sqcup\cdots\sqcup \pi_m$. For partition lattices, the interval factorization implies the M\"obius
factorization
\[
\mu(\pi,\rho)=\prod_{i=1}^m \mu(\pi_i,{\bf 1}_{R_i}),
\]
where ${\bf 1}_{R_i}$ is the one-block partition of $R_i$.
Moreover, by definition of $f_\pi$ we have the blockwise decomposition
\[
f_\pi = \sum_{i=1}^m \sum_{C\in\pi_i} S^{(n)}_C \;=:\; \sum_{i=1}^m f^{(R_i)}_{\pi_i}.
\]
Therefore
\begin{align*}
M_\rho
&=\sum_{\pi\le\rho}\mu(\pi,\rho)\,f_\pi
 =\sum_{\pi_1,\dots,\pi_m}\Big(\prod_{i=1}^m \mu(\pi_i,{\bf 1}_{R_i})\Big)\Big(\sum_{j=1}^m f^{(R_j)}_{\pi_j}\Big)\\
&=\sum_{j=1}^m\left(\sum_{\pi_j}\mu(\pi_j,{\bf 1}_{R_j})\,f^{(R_j)}_{\pi_j}\right)
\prod_{i\neq j}\left(\sum_{\pi_i}\mu(\pi_i,{\bf 1}_{R_i})\right).
\end{align*}
For each $i$ with $|R_i|\ge 2$, the M\"obius identity on the nontrivial interval $[{\bf 0}_{R_i},{\bf 1}_{R_i}]$ gives
\[
\sum_{\pi_i\in\Pi_{R_i}}\mu(\pi_i,{\bf 1}_{R_i})=0.
\]
Since $m\ge 2$, for every $j$ the product over $i\neq j$ contains at least one such zero factor, hence $M_\rho=0$.
This also shows why $M_{\bf 1}$ may be non-zero. In that case $m=1$.

2. For the top element ${\bf 1}$ one has $\mu(\pi,{\bf 1})=(-1)^{|\pi|-1}(|\pi|-1)!$.
Fix a nonempty proper subset $A\subsetneq X$ and collect the coefficient of $S^{(n)}_A$ in $M_{\bf 1}$.
A partition $\pi$ contributes to $S^{(n)}_A$ precisely when $A$ is a block of $\pi$, i.e.\ when $\pi=\{A\}\sqcup \pi'$
for some partition $\pi'$ of $X\setminus A$. If $|X\setminus A|=m$ and $\pi'$ has $r$ blocks, then $|\pi|=r+1$ and
$\mu(\pi,{\bf 1})=(-1)^r r!$. Summing over all $\pi'$ yields the coefficient
\[
\sum_{r=0}^m (-1)^r r!\,S(m,r),
\]
where $S(m,r)$ are Stirling numbers of the second kind. Using the standard identity
$\sum_{r=0}^m (-1)^r r!\,S(m,r)=(-1)^m$, we obtain the expansion:
\[
M_{\bf 1}(|\psi\rangle)=\sum_{\varnothing\neq A\subsetneq X} (-1)^{|X\setminus A|}\,S^{(n)}_A
=\sum_{\varnothing\neq A\subsetneq X} (-1)^{q-|A|}\,\,S^{(n)}_A.
\]

If $|\psi\rangle$ is pure on $X$, then $S^{(n)}_A=S^{(n)}_{A^c}$ for all $A\subseteq X$.
Pairing the terms $A$ and $A^c$, the combined coefficient is
\[
(-1)^{q-|A|}+(-1)^{q-|A^c|}
= (-1)^{|A|}\bigl((-1)^q+1\bigr),
\]
which vanishes when $q$ is odd. Hence $M_{\bf 1}(|\psi\rangle)=0$ for odd $q$ in the pure-state case. For even $q$, using $S_{\varnothing}=S_{X}=0$, we get the formula,
\begin{align}
    M_{\bf 1}= \sum_{A\subset X} (-1)^{q-|A|}\,\,S^{(n)}_A.
\end{align}
\end{proof}

\bibliographystyle{apsrev4-2}
\bibliography{references.bib}

\end{document}